\newif\ifoasics
\oasicstrue
\oasicsfalse

\ifoasics
\documentclass[a4paper,USenglish,numberwithinsect]{oasics-v2018}
\usepackage{microtype}
\else
\documentclass[11pt,fleqn]{article}
\usepackage{fullpage}
\fi

\usepackage{float}
\usepackage{xspace}
\usepackage[ruled,linesnumbered,vlined]{algorithm2e}

\ifoasics

\else
\usepackage{comment}
\usepackage{color,colortbl}
\usepackage{amsthm}
\usepackage{amsmath}
\usepackage{amssymb}
\usepackage{graphicx}
\usepackage{nicefrac}
\usepackage[colorinlistoftodos,prependcaption,textsize=tiny]{todonotes}
\usepackage{caption}
\usepackage{url}
\definecolor{Darkblue}{rgb}{0,0,0.4}
\definecolor{Brown}{cmyk}{0,0.61,1.,0.60}
\definecolor{Purple}{cmyk}{0.45,0.86,0,0}
\fi

\ifoasics
\else
\usepackage[colorlinks,linkcolor=Darkblue,filecolor=blue,citecolor=blue,urlcolor=Darkblue,pagebackref,pdfstartview=FitH]{hyperref}
\fi
\usepackage[nameinlink]{cleveref}
\usepackage{pdfsync}

\ifoasics
\else
\newtheorem{theorem}{Theorem}[section]
\newtheorem{lemma}[theorem]{Lemma}

\fi

\newcommand{\oldparagraph}[1]{\subparagraph{#1}}

\providecommand{\card}[1]{\lvert#1\rvert}

\ifoasics
\else
\newcommand{\namedref}[2]{\hyperref[#2]{#1~\ref*{#2}}}

\fi

\newcommand{\ProblemName}[1]{\textsf{#1}}

\newcommand{\ZEX}{\ProblemName{$0$-Extension}\xspace}

\newcommand{\R}{\mathbb{R}}

\newcommand{\Exp}{\mbox{\rm EXP}}

\ifoasics
\EventEditors{Jeremy Fineman and Michael Mitzenmacher}
\EventNoEds{2}
\EventLongTitle{2nd Symposium on Simplicity in Algorithms (SOSA 2019)}
\EventShortTitle{SOSA 2019}
\EventAcronym{SOSA}
\EventYear{2019}
\EventDate{January 9, 2019}
\EventLocation{Westin San Diego, San Diego, CA, USA}
\EventLogo{}
\SeriesVolume{69}
\ArticleNo{10} 
\else
\fi

\DeclareMathOperator*{\ddim}{ddim}
\newcommand{\eps}{\epsilon}

\newcommand{\la}{~\leftarrow~}

\newcommand{\iid}{i.i.d.\xspace}
\DeclareMathOperator*{\EX}{{\mathbb E}}

\newcommand{\MOE}{\textsf{M$0$E}\xspace}
\newcommand{\SPR}{\textsf{SPR}\xspace}
\newcommand{\NV}{\texttt{Relaxed-Voronoi}\xspace}

\def\inline#1:{\par\vskip 7pt\noindent{\bf #1:}\hskip 10pt}

\def\inline#1:{\par\vskip 7pt\noindent{\bf #1:}\hskip 10pt}

\def\blackslug{\hbox{\hskip 1pt \vrule width 4pt height 8pt
		depth 1.5pt \hskip 1pt}}

\def\QED{\quad\blackslug\lower 8.5pt\null\par}

\newcommand{\initOneLiners}{%
    \setlength{\itemsep}{0pt}
    \setlength{\parsep }{0pt}
    \setlength{\topsep }{0pt}
}

\newcommand{\alert}[1]{\textbf{\color{red}
		[[[#1]]]}\marginpar{\textbf{\color{red}**}}\typeout{ALERT:
		\the\inputlineno: #1}}
\definecolor{purple}{rgb}{0.294, 0, 0.71}

\floatstyle{ruled}
\newfloat{algorithm}{tbp}{loa}
\providecommand{\algorithmname}{Algorithm}
\floatname{algorithm}{\protect\algorithmname}

\ifoasics
\author{Arnold Filtser}
{Ben-Gurion University of the Negev, Be'er Sheva, Israel}
{arnoldf@cs.bgu.ac.il}
{}
{Work partially supported by the Lynn and William Frankel Center for Computer Sciences, ISF grant 1817/17, and by BSF Grant 2015813.}
\author{Robert Krauthgamer}
{Weizmann Institute of Science, Rehovot, Israel}
{robert.krauthgamer@weizmann.ac.il}
{}
{Work partially supported by ONR Award N00014-18-1-2364, the Israel Science Foundation grant \#1086/18, and a Minerva Foundation grant.}
\author{Ohad Trabelsi}
{Weizmann Institute of Science, Rehovot, Israel}
{ohad.trabelsi@weizmann.ac.il}
{}
{Work partly done at IBM Almaden.}

\authorrunning{A. Filtser, R. Krauthgamer and O. Trabelsi}

\Copyright{Arnold Filtser, Robert Krauthgamer and Ohad Trabelsi}

\subjclass{Theory of computation $\rightarrow$ Random projections and metric embeddings, Graph algorithms analysis}


\category{}

\relatedversion{}

\supplement{}

\funding{}

\acknowledgements{}

\nolinenumbers 

\else
\author{Arnold Filtser%
  \thanks{Ben-Gurion University of the Negev. 
  	 Work partially supported by the Lynn and William Frankel Center for Computer Sciences, ISF grant 1817/17,
  	and by BSF Grant 2015813.
    Email: \texttt{arnoldf@cs.bgu.ac.il}
  }
  \and
  Robert Krauthgamer%
  \thanks{Weizmann Institute of Science.
    Work partially supported by ONR Award N00014-18-1-2364, the Israel Science Foundation grant \#1086/18, and a Minerva Foundation grant.
    Email: \texttt{robert.krauthgamer@weizmann.ac.il}
  }
  \and
  Ohad Trabelsi%
  \thanks{Weizmann Institute of Science.
    Work partly done at IBM Almaden.
    Email: \texttt{ohad.trabelsi@weizmann.ac.il}
  }
}
\fi

\title{Relaxed Voronoi: A Simple Framework for Terminal-Clustering Problems%
  \footnote{In earlier versions this algorithm was called ``Noisy Voronoi''.}
}

\begin{document}

\maketitle

\begin{abstract}
We reprove three known algorithmic bounds for terminal-clustering problems, 
using a single framework that leads to simpler proofs.
In this genre of problems, the input is a metric space $(X,d)$ 
(possibly arising from a graph) and a subset of terminals $K\subset X$, 
and the goal is to partition the points $X$ such that each part, called a cluster, contains exactly one terminal 
(possibly with connectivity requirements) so as to minimize some objective.
The three bounds we reprove are for 
Steiner Point Removal on trees [Gupta, SODA 2001], 
for Metric $0$-Extension in bounded doubling dimension [Lee and Naor, unpublished 2003], 
and for Connected Metric $0$-Extension [Englert et al., SICOMP 2014]. 

A natural approach is to cluster each point with its closest terminal,
which would partition $X$ into so-called Voronoi cells,
but this approach can fail miserably due to its stringent cluster boundaries. 
A now-standard fix, which we call the \texttt{Relaxed-Voronoi} framework, is to use enlarged Voronoi cells, but to obtain disjoint clusters, the cells are computed greedily according to some order.
This method, first proposed by Calinescu, Karloff and Rabani [SICOMP 2004], 
was employed successfully to provide state-of-the-art results
for terminal-clustering problems on general metrics. 
However, for restricted families of metrics, e.g., trees and doubling metrics, 
only more complicated, ad-hoc algorithms are known.  
Our main contribution is to demonstrate that the \texttt{Relaxed-Voronoi} algorithm is applicable to restricted metrics, 
and actually leads to relatively simple algorithms and 
analyses. 
\end{abstract}

\section{Introduction}\label{sec:intro}

We consider \emph{terminal clustering} problems, 
where the input is a metric space $(X,d)$ with $k$ terminals $K\subseteq X$, and the goal is to partition the points (vertices) into $k$ clusters,
each containing exactly one terminal, so as to minimize some objective. 
In the \emph{graphical version} of this problem, 
the input is a weighted graph $G=(V,E,w)$ with terminals $K\subset V$
and the metric $d$ is derived as the shortest-path metric on $X=V$
with respect to the non-negative edge weights $w$, 
and every output cluster should be connected (as an induced subgraph of $G$). 

We present for these problems a simple algorithmic framework 
that generalizes two different known algorithms, from \cite{CKR04,Fil18a}.
Using this framework, we obtain simple algorithms 
for two specific metric/graph classes, 
and recover their known bounds from \cite{G01,LN03,EGKRTT14} in a unified manner that is arguably simpler and more insightful than previous work.
In our case, even the analysis is short and simple. 
Thus, our main contribution is to identify and present the framework,
and to (non-trivially) apply it to specific metric/graph classes, and we hope it will lead to new results in the future.
We proceed to define the two specific problems that we investigate,
and briefly survey their known bounds.

\oldparagraph{Metric $0$-Extension (\MOE)} 
In this problem, 
the input is a metric space $(X,d)$ and a set of $k$ terminals $K\subset X$,
and the goal is to find a distribution $\mathcal{D}$ over retractions $f$
(i.e., functions $f:X\rightarrow K$ that satisfy $f(x)=x$ for all $x\in K$), 
such that 
\[
\forall x,y\in X,
\qquad 
\EX_{f\sim\mathcal{D}}\left[d(f(x),f(y))\right]\le\alpha\cdot d(x,y), 
\]
where $\alpha\ge1$, called the \emph{expected distortion}, 
is as small as possible.
Throughout, we seek the smallest $\alpha$ that holds for a class of metric spaces, 
for example all metrics with $k$ terminals, and then $\alpha=\alpha(k)$. 

The above is closely related to the well-known \emph{\ZEX} problem,
in which the input is a set $X$, a terminal set $K\subseteq X$, a metric $d_K$ over the terminals and a cost function $c:{X\choose 2}\rightarrow \mathbb{R}_+$,
and the goal is to find a retraction $f:X\rightarrow K$ 
that minimizes $\sum_{\{x,y\}\in {X\choose 2}} c(x,y)\cdot d_K(f(x),f(y))$.
The \ZEX problem, first proposed by Karzanov~\cite{Kar98}, 
generalizes the \ProblemName{Multiway Cut} problem \cite{DJPSY92} 
by allowing $d_K$ to be any discrete metric (instead of a uniform metric) and it is 
also a special case of the \ProblemName{Metric Labeling} problem \cite{KT02}, whose objective function has additional terms 
that represent assignment costs.
Karzanov introduced a linear programming (LP) relaxation for \ZEX,
which can be described as finding a (semi-)metric 
$d_X$ over $X$ that agrees with $d_K$ on $K$, and minimizes $\sum_{\{x,y\}\in {X\choose 2}} c(x,y)\cdot d_X(x,y)$. 
Rounding this LP relaxation is equivalent to the \MOE problem
(by the minimax theorem). 
Consequently, most previous work on \ZEX has actually focused on solving \MOE, 
and so does our work.

A well-known open problem is to determine the smallest distortion $\alpha(k)$ 
that suffices for all metric spaces with $k$ terminals. 
The currently known bounds are $O\left(\log k/\log\log k\right)$ 
due to Fakcharoenphol, Harrelson, Rao, and Talwar~\cite{FHRT03} 
(improving over \cite{CKR04}), 
and $\Omega(\sqrt{\log k})$ due to Calinescu, Karloff and Rabani~\cite{CKR04}. 
Improved upper bounds are known for special classes of metric spaces $X$,
for example $O(1)$ for the case where $X$ is 
the shortest-paths metric of a graph excluding a fixed minor \cite{CKR04}.
Another example is when the submetric on the terminals 
(i.e., the restriction of $d$ to $K$) is $\beta$-decomposable, 
which admits an $O(\beta)$ upper bound \cite{LN03} 
(a somewhat similar bound was obtained in \cite{AFHKTT04}). 
This implies an $O(\ddim(K))$ upper bound,
where $\ddim(K)$ denotes the doubling dimension of the terminals' submetric
(see \Cref{sec:prelims} for definition),
and our results reproduce the latter bound.

\oldparagraph{Steiner Point Removal (\SPR)}
In this problem, given a weighted graph $G=(V,E,w)$ 
and a set of terminals $K\subseteq V$,
the goal is to find a minor $M=(K,E')$ of $G$
(note its vertex set is exactly the set of terminals), 
that approximately preserves the distances between terminals,
which means (using $d_H$ to denote the shortest-path metric in $H$) that 
\[
\forall t,t'\in K, 
\qquad
d_G(t,t')\le d_{M}(t,t')\le \alpha \cdot d_G(t,t') ,
\] 
where $\alpha\ge1$, called the \emph{distortion}, is as small as possible.
Again, we seek the best $\alpha$ that holds for a class of graphs,
say all graphs with $k=\card{K}$ terminals. 

Let us denote $K=\{t_1,\dots,t_k\}$. 
A partition $\{V_1,\dots,V_k\}$ of $V$ is called a \emph{terminal partition} (with respect to $K$) if for all $i=1,\ldots,k$, 
the induced subgraph $G[V_i]$ is connected and contains $t_i$. 
The \emph{induced minor} $M$ of such a terminal partition 
is the minor obtained by contracting each $V_i$ into a single vertex called (abusing notation) $t_i$. 
Thus, $M$ has an edge between $t_i$ and $t_j$ iff $G$ has an edge 
between $V_i$ and $V_j$. 
The weight of this edge (if exists) is simply $d_G(t_i,t_j)$,
which represents the shortest-path in $G$; 
see \Cref{fig:contraction} for an example.
Most of the work on \SPR so far used terminal partitions to obtain a minor, 
and so does our work.

\begin{figure}[ht]
  \centering{\includegraphics[scale=0.75]{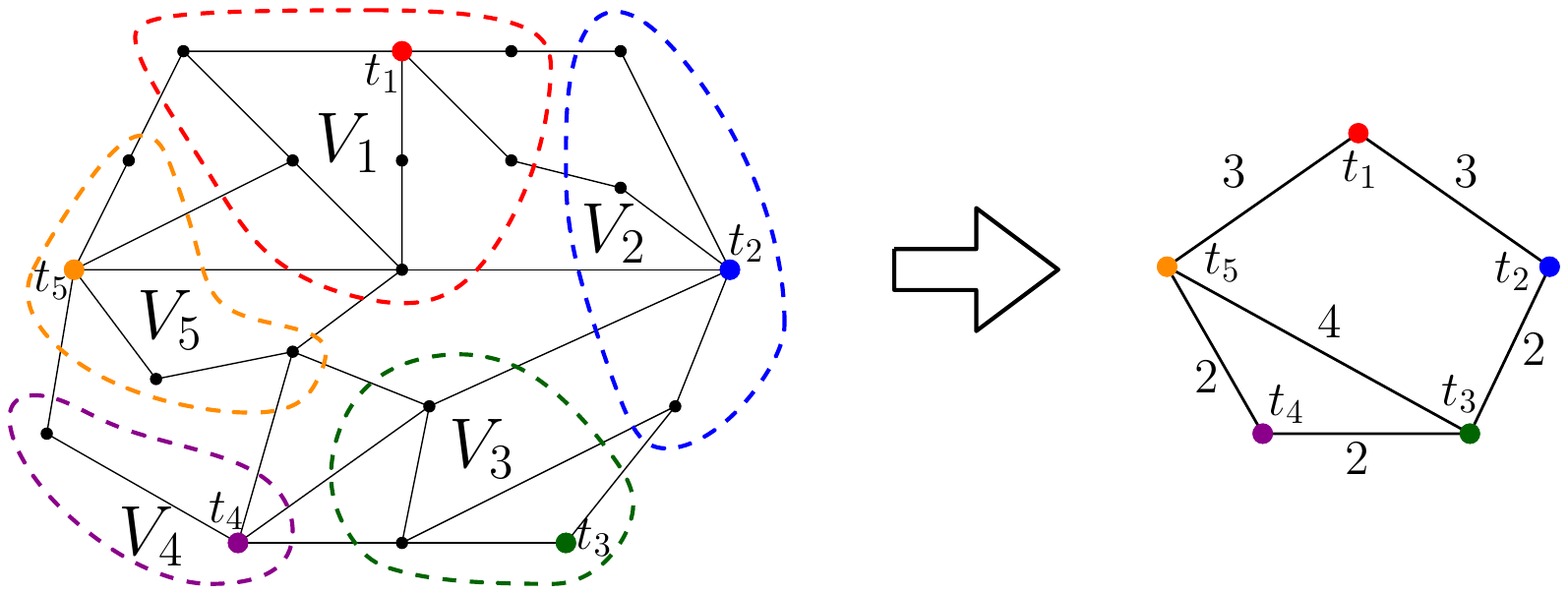}} 
  \caption{\label{fig:contraction}\small \it
    Example how a terminal partition of graph $G$ (on left) induces a minor $M$ (on right). The graph shown has unit weight edges and $5$ terminals, and the terminal partition is shown using dashed curves.
    The distortion is 
    $\frac{d_M(t_2,t_5)}{d_G(t_2,t_5)}=\tfrac{6}{2}=3$.
  }
\end{figure}

For the case where the graph $G$ is a tree, 
the smallest distortion possible for \SPR is known to be $8$. 
Gupta~\cite{G01} constructed a tree achieving distortion $8$; 
in fact, he was only interested in constructing a tree with vertex set $K$,  and later Chan, Xia, Konjevod, and Richa \cite{CXKR06} 
observed that Gupta's tree is actually a minor of the given tree $G$.
Surprisingly, they further showed that $8$ 
is the best possible distortion for the family of trees, 
as (unweighted) complete binary trees require distortion $8-\epsilon$. 
Our results reproduce this upper bound of $8$. 

For SPR in general graphs there is currently a huge gap. 
The best lower bound known is just $8$, known for trees,
and recently Filtser \cite{Fil18} showed an $O(\log k)$ upper bound
(improving over \cite{KKN15,Che18}). 
No better upper bound is known even for seemingly much simpler cases such as planar graphs,
and the only other bound known is $\alpha=O(1)$ for outerplanar graphs
\cite{BG08}. 

\subsection{Algorithmic Framework} 
\label{sec:alg}
A natural and straightforward algorithm for terminal clustering
is to simply partition the metric (or graph) into Voronoi cells,
i.e., map each point (or vertex) to its closest terminal,
to obtain a partition of $X$ (or $V$) with one cluster for each terminal. 
However, there are easy examples where this algorithm fails miserably,
because of the stringent cluster boundaries.
A now-standard fix is to build around each terminal (iteratively) 
a cluster that is an enlarged Voronoi cell in the remaining metric (or graph).

\begin{algorithm}[p]
	\caption{\texttt{Metric-Relaxed-Voronoi}}	\label{alg:metric}
	\DontPrintSemicolon
	\SetKwInOut{Input}{input}\SetKwInOut{Output}{output}
	\Input{metric $M=(X,d)$, terminals $K$, ordering $\pi=(t_1,\dots,t_k)$, \\magnitudes $R_1,\dots,R_k\ge 1$}
	\Output{retraction $f:X\rightarrow K$ (i.e., $\forall x\in K$, $f(x)=x$)}
	\BlankLine
	\For {$j=1,\ldots,k$}{
		\For {all unmapped points $x$ such that $d(t_{j},x)\le R_{j}\cdot D(x)$}{
			set $f(x)=t_j$\;
		}	
	}
	\Return $f$\;
\end{algorithm}

This approach was first used by Calinescu, Karloff and Rabani~\cite{CKR04}. 
We generalize their method, so that all previous uses of this approach
can be seen as instantiations of our algorithm with specific parameters. 
Our algorithm, called \texttt{Relaxed-Voronoi},
is formally described in Algorithm \ref{alg:metric}
where throughout we define 
$$
D(x)=d(x,K) = \min_{t\in K} d(x,t)
$$ 
to be the distance from $x\in X$ to its closest terminal.
The algorithm's parameters, formally presented as part of the input, 
are an ordering $\pi=(t_1,\dots,t_k)$ of the terminals 
and corresponding magnitudes $R_1,\ldots,R_k\ge 1$ (one for each terminal). 
The algorithm is rather simple;
each terminal $t_j$, in turn according to the ordering,
creates a cluster $V_j=f^{-1}(t_j)$ containing all yet-unclustered points $x$ 
at distance $d(x,t_j)\le R_j\cdot D(v)$.
That is, the cluster of $t_j$ is a Voronoi cell ``enlarged'' by factor $R_j$ 
in the remaining metric. 
Setting $R_1=\cdots=R_k=1$ recovers the partition into Voronoi cells.

The above algorithm cannot be used as is for the \SPR problem, 
because a terminal partition has an additional connectivity requirement.
Therefore, in the graphical case, 
instead of taking all remaining vertices $x$ that satisfy 
$d_G(x,t_j)\le R_j\cdot D(v)$, 
we create $V_j$ in a Dijkstra-like iterative fashion, as follows. 
Initially $V_j=\{t_j\}$, 
and we repeatedly add to $V_j$ any unclustered vertex 
that has a neighbor in $V_j$ and is at distance $d_G(v,t_j)\le R_j\cdot D(v)$. 
See Algorithms~\ref{alg:mainSPR} and~\ref{alg:CreateCluster} for a formal description. 
This version of the \texttt{Relaxed-Voronoi} algorithm was first proposed 
by Filtser \cite{Fil18a} for the \SPR problem in general graphs. 
It is simpler to describe and to analyze than the 
\texttt{Ball-Growing} algorithm of previous work \cite{KKN15,Che18,Fil18}.%
\footnote{The \texttt{Ball-Growing} algorithm creates clusters in rounds, where each round iteratively enlarges every cluster,
	by increasing its radius around each terminal (in the remaining graph)
	by a value sampled from an exponential distribution.}
Filtser also showed that the \NV algorithm can be implemented in time 
$O(|E|\log |V|)$.%
\footnote{\label{footnote:RealAlg}The $O(|E|\log |V|)$-time in \cite{Fil18a}
	actually implements a slightly different algorithm,
	where the test $d_{G}(v,t_j) \le R_j\cdot D(v)$  (line \ref{line:desidion}) 
	is replaced by $d_{G[V_j\cup\{v\}]}(v,t_j) \le R_j\cdot D(v)$. 
	The distortion bound holds for this algorithm too.
}

\begin{algorithm}[p]
	\setcounter{AlgoLine}{0}
	\caption{\texttt{Graphic-Relaxed-Voronoi}}\label{alg:mainSPR}
	\DontPrintSemicolon
	\SetKwInOut{Input}{input}\SetKwInOut{Output}{output}
	\Input{weighted graph $G=(V,E,w)$, terminals $K$, ordering $\pi=(t_1,\dots,t_k)$, \\magnitudes $R_1,\dots,R_k\ge 1$}
	\Output{Minor $M$}
	\BlankLine
	$V_\perp~\la~V\setminus  K$\hspace{100pt}\emph{//~$V_\perp$ is the currently unclustered vertices.}\;
	
	\For {$j=1,\dots,k$}{
		$V_j\la \texttt{Create-Cluster}(G,V_\perp ,t_j,R_j)$\;
		
		$V_\perp\leftarrow V_\perp\setminus V_j$\;
	}	
	\Return the terminal-centered minor $M$ of $G$ induced by $V_1,\ldots,V_k$\;
\end{algorithm}

\begin{algorithm}[p]
	\setcounter{AlgoLine}{0}
	\caption{\FuncSty{\texttt{Create-Cluster}}}
	\label{alg:CreateCluster}
	\DontPrintSemicolon
	\SetKwInOut{Input}{input}\SetKwInOut{Output}{output}
	\Input{weighted graph $G=(V,E,w)$, unclustered vertices $V_\perp$, terminal $t_j$, magnitude $R_j$}
	\Output{cluster $V_j$}
	\BlankLine
	$V_j\leftarrow \{t_j\}$, $U\la \emptyset$, $N\leftarrow \{\text{all neighbors of $t_j$ in $V_\perp$} \}$\;
	
	\While{$N\ne \emptyset$}{
		pick an arbitrary vertex $v\in N$ and remove it from $N$\;
	
		\If{$d_{G}(v,t_j)\le R_j\cdot D(v)$\label{line:desidion}}{
			add $v$ to $V_j$\;
			
			add all the neighbors of $v$ in $V_\perp\setminus \left(U\cup V_j\right)$ to $N$ \;
		}
		\Else{
			add $v$ to $U$\;
		}
	}
	\Return $V_j$\;
\end{algorithm}

\subsection{Our Contribution}
\label{sec:contribution}

All previous uses of the \NV algorithm were on general metrics or graphs.
Specifically, 
Calinescu et al.~\cite{CKR04} and Fakcharoenphol et al. \cite{FHRT03},
used a uniformly random ordering $\pi$ 
and a single random magnitude $R$ (same for all terminals),
and Filtser~\cite{Fil18a} used an arbitrary ordering $\pi$ 
and magnitudes that are independently and identically distributed (\iid) 
drawn from an exponential-like distribution.
However, for special families of metrics or graphs,
this type of algorithm was never used;
instead, ad-hoc algorithms were developed,
leading to more involved algorithms and analyses. 
Our contribution is to tailor the \NV algorithm to special input families
by choosing the ordering $\pi$ deterministically but depending on the input 
at hand (rather than a random or arbitrary ordering).
As a result, we reprove three known results using simpler algorithms and
analyses. 
We believe that this approach will lead to additional and new results.

\oldparagraph{\SPR on Trees} 
Gupta's algorithm \cite{G01}, which achieves distortion $8$, 
is designed specifically for trees and it is unclear how to generalize it. 
Its recursive definition makes it arguably difficult to understand intuitively
how its output on a given tree would look like. 
For example, the fact that the algorithm is tight and produces a minor \cite{CXKR06} was non-trivial and even surprising. 
This result has proved useful in the past, 
yet it is a bit mysterious why $8$ is the optimal bound,
i.e., what tradeoff does it optimize.

We use the \NV algorithm to construct a tree with optimal distortion $8$. 
The choice of parameters in the algorithm is very simple --- 
the magnitudes are all set to $R_j=3$, 
and the ordering $\pi$ is defined by listing the terminals 
in order of increasing distance from an arbitrary ``root'' vertex $v$
(breaking ties arbitrarily).
Our algorithm's description is simple and intuitive, 
its distortion bound $8$ is explained by the analysis,
and it is straightforward that the output tree is a minor of the input tree. 
Perhaps surprisingly, our algorithm outputs the same tree as Gupta's algorithm. 
Overall, our algorithm provides a better understanding of Gupta's celebrated result. 
We believe that this approach can be generalized to additional graph families, and hopefully achieve a constant distortion for \SPR on (say) planar graphs
(where the current bound is only $O(\log k)$, which holds for general graphs).

\oldparagraph{\MOE on Doubling Metrics} 
Lee and Naor's \cite{LN03} algorithm achieves $O(\ddim)$ 
when the submetric on the terminals 
(i.e., the metric's restriction to points in $K$)
has doubling dimension at most $\ddim$. 
Their algorithm is based on stochastic decompositions,
specifically converting padded decompositions into separating decompositions, then defining (new) partial decompositions, 
and finally using these decompositions in all the possible distance scales.

We use the \NV algorithm to achieve the same $O(\ddim)$ upper bound, 
by setting the parameters as follows. 
The magnitudes $R_j$ are \iid, each distributed like $2\cdot e^{Z}$ where 
$Z$ is drawn from an exponential distribution with parameter $\Theta(\ddim)$. 
We set $\pi$ to be the Gonzalez order~\cite{G85}, 
where $t_1$ is an arbitrary terminal,
and each successive $t_i$ is the terminal farthest from $\{t_1,\dots,t_{i-1}\}$,
breaking ties arbitrarily.
Our algorithm is much simpler, more elegant, and its straightforward implementation takes only $O(nk)$ time (assuming the input is given as a matrix of pairwise distances). 
We hope that our ideas could lead to a better upper bound for the \SPR problem in the case where the metric restricted to the terminals has a bounded doubling dimension.

\oldparagraph{Connected \MOE} 
This is a graphic version of the \MOE problem. The input metric 
is the shortest-path metric of an edge-weighted graph $G=(V,E,w)$,
and similarly to the \MOE problem, the goal is to find a distribution over retractions $f:V\rightarrow K$, but with an additional requirement: 
each cluster $f^{-1}(t_j)$ must be connected (as a subgraph of $G$).
Englert et al.~\cite{EGKRTT14} achieved for this problem 
expected distortion $\alpha=O(\log k)$
using an algorithm that partitions the graph vertices into clusters
using stochastic decompositions in all possible distance scales, 
and then merging some clusters to enforce connectivity.
We use a graphic version of the \NV algorithm (which guarantees connectivity) 
to achieve the same expected distortion $O(\log k)$. 
When describing this algorithm, we abuse notation and identify $f(v)=t_j$ 
with $v\in V_j$, i.e., when the algorithm adds a vertex $v$ to cluster $V_j$, it should be understood as also assigning $f(v)=t_j$. 
The graphic \NV algorithm is much simpler than the previous algorithm
of \cite{EGKRTT14}, and we set its parameters as follows.
The ordering $\pi$ is arbitrary,
and the magnitudes $R_j$ are \iid, each distributed like $e^{Z}$ where 
$Z$ is drawn from an exponential distribution with parameter $\Theta(\log k)$.
Even though this problem is concerned with general graphs 
and there is nothing clever about the ordering, 
we still chose to present this result, as it gives further evidence
to the strength and broad applicability of the \NV algorithm. 
Another advantage is that it can be implemented in $O(|E|\log |V|)$ time, 
while the algorithm of~\cite{EGKRTT14} requires more time (an unspecified 
polynomial). See \Cref{footnote:RealAlg} for additional details.

\subsection{Related Work}
The Voronoi-like approach was used also in other recent algorithms.
Gupta and Talwar~\cite{GT13} introduced the \texttt{Random-Rates} algorithm,
in which each terminal $t_j$ samples a rate $\rho_j\ge 1$,
and then every point $x$ is clustered with the terminal $t_j$ that minimizes
the ratio $\frac{d(x,t_j)}{\rho_j}$.
The main difference from the \NV algorithm is that in their algorithm,
the terminals create their clusters simultaneously (rather than sequentially),
which does not guarantee that the clusters are connected.
Gupta and Talwar~\cite{GT13} proved an $O(\log k)$ expected distortion
for this algorithm on the \MOE problem.
It seems unlikely that their algorithm can provide $O(\ddim(K))$ upper bound,
which usually follows by bounding the number of clusters relevant
to any ``separation event'' by $2^{O(\ddim(K))}$. 
We achieve this using the sequential ordering, 
but in their algorithm too many clusters can be relevant.

Miller, Peng and Xu~\cite{MPX13} introduced the \texttt{Parallel-Partition} algorithm to partition a graph into low-diameter clusters
(without a given set of terminals).
In this algorithm, each vertex $u$ samples a random shift $s_u\ge 0$, 
and then every vertex $x$ joins the cluster of $u$ with minimum $d(x,u)-s_u$.
This algorithm produces connected clusters,
however, it gets as an input a target diameter $\Delta>0$,
and its guarantees are proportional to this parameter. 
In contrast, the \NV algorithm is scale-free and handles all distances scales simultaneously (similar to the above \texttt{Random-Rates} algorithm),
and therefore it is more natural for terminal-partitioning problems.

\section{Preliminaries}
\label{sec:prelims}
Consider an undirected graph $G=(V,E)$ with non-negative edge weights
$w: E \to \R_{\geq 0}$ and let $d_{G}$ denote the shortest-path metric in $G$. 
For a subset of vertices $A\subseteq V$, let $G[A]$ denote the \emph{induced graph} on $A$. 
Fix $K=\{t_1,\dots,t_k\}\subseteq V$ to be a set of the given \emph{terminals}. 
As mentioned earlier, 
for a vertex $v\in V$ we define $D(v)=\min_{t\in K}d_{G}(v,t)$ 
to be the distance from $v$ to its closest terminal. 

A graph $H$ is a \emph{minor} of a graph $G$ if it can be obtained from $G$ 
by edge deletions, edge contractions, and vertex deletions. 
As defined earlier, a partition $\{V_1,\dots,V_k\}$ of $V$ 
is called a \emph{terminal partition} 
(with respect to $K$) if for all $i=1,\ldots,k$, 
the induced subgraph $G[V_i]$ is connected and contains $t_i$. 
The \emph{minor induced} by a terminal partition $\{V_1,\dots,V_k\}$ 
is the minor $M$ obtained by contracting each set $V_i$ 
into a single vertex called (abusing notation) $t_i$. 
Notice that $M$ has an edge between $t_i$ and $t_j$ iff there are vertices $v_i\in V_i$ and $v_j\in V_j$ such that $\{v_i,v_j\}\in E$. 
The weight of this edge (if exists) is simply $d_G(t_i,t_j)$,
which represents the shortest-path in $G$. 
It is easily verified that by the triangle inequality, 
for every pair of (not necessarily adjacent) terminals $t_i,t_j$, 
we have $d_M(t_i,t_j)\ge d_G(t_i,t_j)$.
The \emph{distortion} of the induced minor is  
$\max_{i\ne j}\frac{d_M(t_i,t_j)}{d_G(t_i,t_j)}$.
It was proved in~\cite{Fil18a} that the \NV algorithm always returns a terminal partition. 
\begin{lemma}[Lemma 2 in \cite{Fil18a}]\label{lem:AlgRetPar}
	The sets $V_1,\ldots,V_k$ constructed by Algorithm \ref{alg:mainSPR} constitute a terminal partition.
\end{lemma}

We say that a metric $(X,d)$ has \emph{doubling dimension} $\ddim$ if every ball of radius $r>0$ can be covered by at most $2^{\ddim}$ balls of radius $r/2$. 
We will use the following \emph{packing property} of doubling spaces \cite{GKL03}: 
Consider a set $N$ such that for every $x\neq y\in N$ it holds that $d(x,y)\ge \delta$. Then every ball of radius $\Delta\ge \delta$ contains at most $\big(\frac{4\Delta}{\delta}\big)^{O(\ddim)}=2^{O\left(\ddim\cdot\log\frac\Delta\delta\right)}$ points from $N$.

We denote by $\Exp(\lambda)$ the \emph{exponential distribution} 
with mean $\lambda>0$, which has density function $f(x)=\frac{1}{\lambda}e^{-\frac{x}{\lambda}}$ for $x\ge0$. 
This distribution is \emph{memoryless}: if $X\sim\Exp(\lambda)$, then for all $a,b\ge 0$ we have $\Pr[X\ge a+b\mid X\ge a]=\Pr[X\ge b]$. 
In other words, conditioned on $X\ge a$, it holds that $X\sim a+\Exp(\lambda)$.

\section{\SPR on trees}
In this section we analyze the \NV algorithm (Algorithm \ref{alg:mainSPR}) on trees.

\begin{theorem}\label{thm:mainSPR}
	Let $T$ be a tree and $r$ be an arbitrary vertex. Let $\pi$ be an ordering of the terminals according to an increasing distance from $r$. Then the tree $T_K$ returned by the \NV algorithm on input $(T,K,\pi,\{3,3,\dots,3\})$ has distortion at most $8$. Moreover, the algorithm can be implemented in linear time.
\end{theorem}

In \Cref{sec:distortion} we bound the distortion produced by our algorithm, and in \Cref{sec:time} we describe its linear time implementation.
See \Cref{fig:btreeExample} for an example execution of the algorithm on a complete unweighted binary tree (the lower bound example used by \cite{CXKR06}).
\begin{figure}[ht]
  \centering{\includegraphics[scale=0.6]{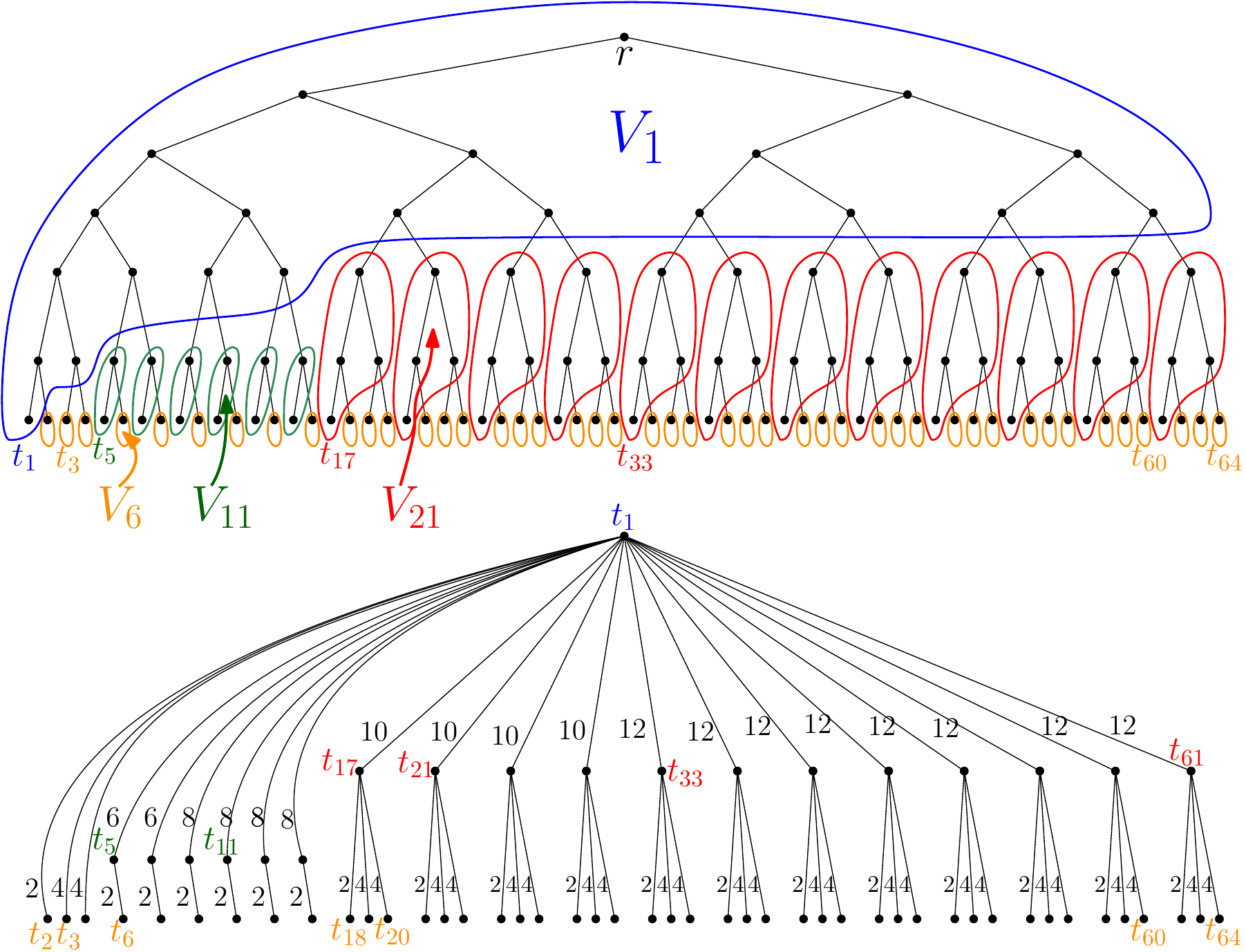}} 
  \caption{\label{fig:btreeExample}\small \it
    An example execution of the \texttt{Relaxed-Voronoi} algorithm.
    The top graph is the input, a complete binary tree of height $6$ with all the leaves as terminals. Choosing $r$ to be the root of the tree, the terminal ordering $\pi$ can be arbitrary, and we choose a left-to-right ordering. The resulting clusters are shown using colored curves. 
    The bottom graph shows the minor induced by the terminal partition above, representing every cluster by its top-most vertex (rather than the terminal).
    The distortion is 
    $\frac{d_{T_k}(t_{64},t_{60})}{d_T(t_{64},t_{60})}=\tfrac{32}{6} = 5\tfrac23$.
  }
\end{figure}

\subsection{Distortion Analysis}\label{sec:distortion}

To better understand the final distortion bound $8$,
we analyze the \texttt{Relaxed-Voronoi} algorithm for a general $R>1$, 
and we optimize it only at the very end, 
setting $R=3$ to obtain distortion $8$. 

Denote by $T_K$ the tree minor returned by the algorithm,
and call the vertex $r$ (used to determine $\pi$) the root. 
Let $t_1$ be the first terminal w.r.t $\pi$,
and let $V_1$ be the cluster that the \texttt{Relaxed-Voronoi} algorithm constructs for $t_1$.
This terminal $t_1$ is the closest terminal to the root $r$,
and actually also to every vertex on the shortest path from $t_1$ to $r$.  Therefore $r$ joins the cluster $V_1$.
Let $C_1,\dots,C_s$ be the connected components of the remaining graph
$G\setminus V_1=G[V\setminus V_1]$,
and let $K_i=C_i\cap K$ be the subset of terminals in component $C_i$. 
We claim that for every vertex $v\in C_i$ (for every $i$),
its closest terminal $t_v$ satisfies $t_v\in K_i$.
Indeed, assume towards contradiction that some vertex $u$ on the path between $v$ to $t_v\in K_i$ joined $V_1$.
Consider then an arbitrary vertex $u'$ on the path from $v$ to $u$,
and note that $t_v$ is also the closest terminal to both $u,u'$.
By the triangle inequality, 
$d_{T}(t_{1},u')\le d_{T}(t_{1},u)+d_{T}(u,u')\le R\cdot\big(d_{T}(t_{v},u)+d_{T}(u,u')\big)=R\cdot d_{T}(t_{v},u')$.
This implies that every vertex on the path from $u$ to $v$ will join $V_1$
(recall we assumed $u$ joins $V_1$, and the algorithm iteratively adds
neighbor of vertices already in $V_1$),
in contradiction with $v\in C_i$.
See \Cref{fig:connnectedT_i} for illustration.
\begin{figure}[ht]
  \centering{\includegraphics[scale=0.9]{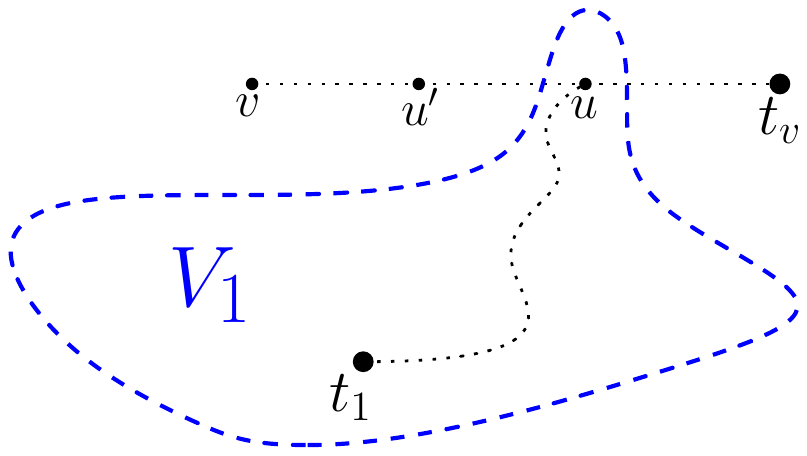}} 
  \caption{\label{fig:connnectedT_i} \small \it 
    Illustrating the argument that for every $v\in C_i$,
    also its closest terminal $t_v\in C_i$.
    Assuming some $u$ on the path between them joined $V_1$,
    we conclude the entire path from $u$ to $v$ joins $V_1$.
  }
\end{figure}

By construction,
there is only a single edge $e_i$ that connects $C_i$ and $V_1$,
and denote its two endpoints by $r_i\in C_i$ and $s_i\in V_1$. 
Let $\tilde{t}_i\in K$ be the closest terminal to $r_i$, thus $\tilde{t}_i\in C_i$. 
Observe that $r_i$ is the closest vertex to $r$ among all vertices in $C_i$,
and in particular every path from a terminal $t\in C_i$ to $r$ goes through $r_i$. Let $\pi_i$ be an ordering of the terminals in $K_i$ according to increasing distance from $r_i$. Note that $\pi_i$ is just the order $\pi$ restricted to $K_i$.
Since all the clusters created by the \texttt{Relaxed-Voronoi} algorithm are connected, no vertex in $C_i$ can join a cluster associated with a terminal outside $K_i$. 
In particular, for every $v\in C_i$ the distance $D(v)$ to the closest terminal in the restricted tree $G[C_i]$ remains the same (as $t_v\in C_i$).
Therefore, if we execute the \texttt{Relaxed-Voronoi} algorithm on $C_i$ with terminal set $K_i$ and order $\pi_i$, the partition of $C_i$ to clusters will be identical to the partition of $C_i$ induced by the original algorithm (on $T$ with the order $\pi$).
Accordingly, if we combine all the clusters created by such executions with $V_1$, we get the same terminal partition as produced by the \texttt{Relaxed-Voronoi} algorithm on the original graph.

Next, we argue by induction on the number of terminals that for every terminal $t$,
$d_{T_K}(t_{1},t)\le \frac{R+1}{R-1}\cdot d_{T}(t_{1},t)$. In a tree with a single or
two terminals this claim is trivial. We now prove the induction step. 
Let $t_i$ be some terminal which belongs to the connected component $C_i$ (in $T\setminus V_1$). By applying the induction hypothesis to the tree $C_i$ with order $\pi_i$, it holds that $d_{T_k}(t_i,\tilde{t}_i)\le \frac{R+1}{R-1}\cdot d_{T}(t_i,\tilde{t}_i)$ as $\tilde{t}_i$ is the first terminal in the order $\pi_i$.
Note that $r_i$ will necessarily join the cluster of $\tilde{t}_i$, therefore the edge ($e_i=\{s_i,r_i\}$) crosses the clusters of $t_1$ and  $\tilde{t}_i$, which implies that there is an edge between $t_1$ to $\tilde{t}_i$ in $T_K$. See \Cref{fig:distortionBound} for illustration.
\begin{figure}[ht]
	\centering{\includegraphics[scale=0.6]{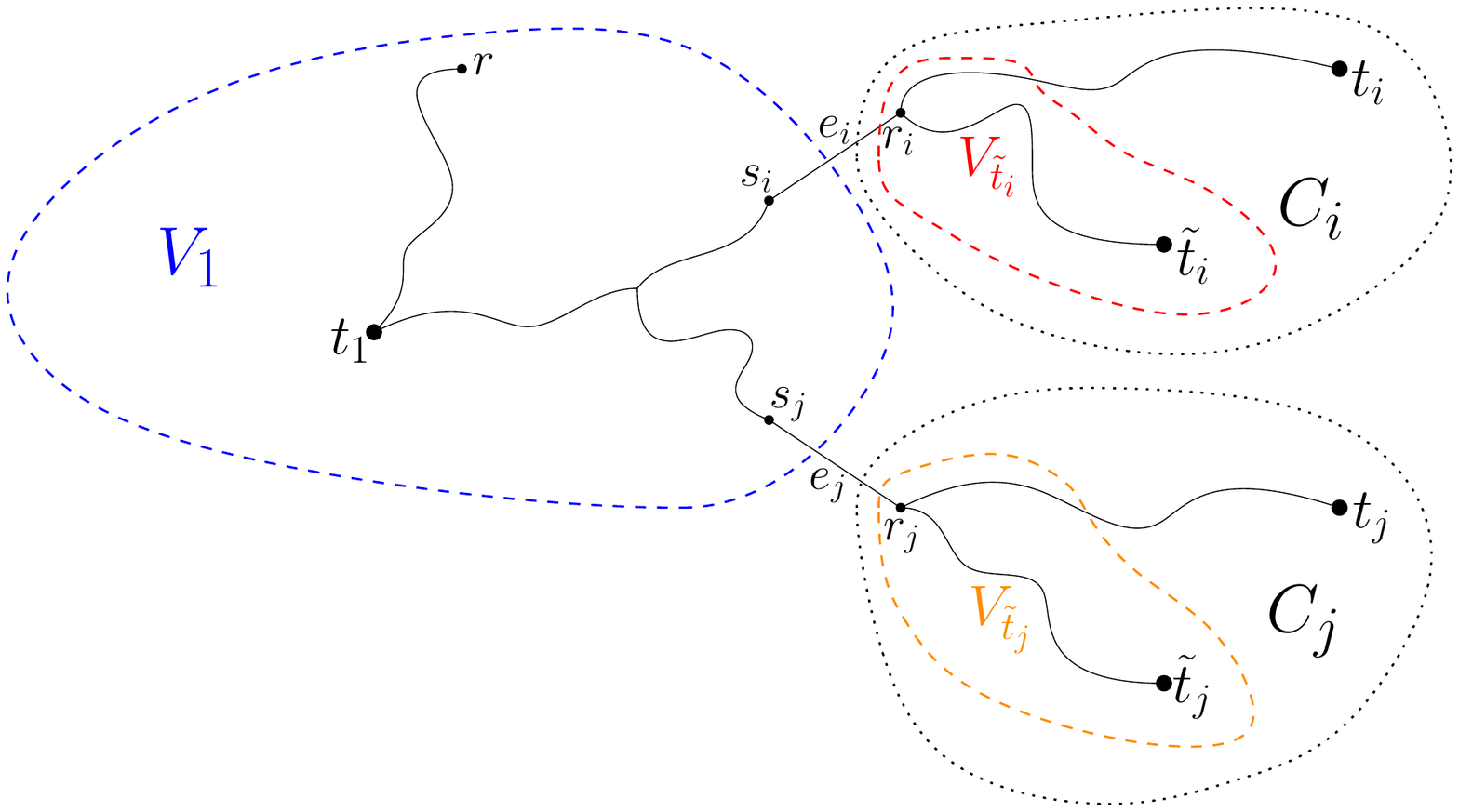}} 
	\caption{\label{fig:distortionBound}\small \it 
	Illustrating the bound on $d_{T_k}(t_i,t_j)$.
Initially $d_{T_k}(t_1,t_i)$ is bounded. Notice that $\tilde{t}_i$ is the closest terminal to $r_i$. Using the induction hypothesis we have that $d_{T_k}(\tilde{t}_i,t_i)\le \frac{R+1}{R-1}\cdot d_T(\tilde{t}_i,t_i)$. As $\{t_1,\tilde{t}_i\}$ is an edge in $T_k$, the bound follows.\\	
Next, the bound on $d_{T_k}(t_i,t_j)$. Notice that $d_T(t_i,t_j)\geq d_T(t_i,s_i)+d_T(t_j,s_j)$. $d_{T_k}(t_i,t_j)$ is upper bounded by going through $t_1$, using the assertion above.}
\end{figure}

As $r_i$ has a neighbor in $V_1$ but did not join $V_1$, necessarily $d_T(t_1,r_i)>R\cdot D(r_i)=R\cdot d_T(r_i,\tilde{t}_i)$. We conclude,
\begin{align*}
d_{T_{K}}\left(t_{1},t_{i}\right) & \le d_{T_{K}}\left(t_{1},\tilde{t}_{i}\right)+d_{T_{K}}\left(\tilde{t}_{i},t_{i}\right)\\
& \le d_{T}\left(t_{1},\tilde{t}_{i}\right)+\frac{R+1}{R-1}\cdot d_{T}\left(\tilde{t}_{i},t_{i}\right)\\
&\le d_{T}\left(t_{1},r_{i}\right)+D(r_i)+\frac{R+1}{R-1}\cdot \big(D(r_i)+d_T(r_i,t_i)\big)\\
&< d_{T}\left(t_{1},r_{i}\right)+\big(1+\frac{R+1}{R-1}\big)\cdot \frac{d_T(t_1,r_i)}{R}+\frac{R+1}{R-1}\cdot d_T(r_i,t_i)\\
&=\frac{R+1}{R-1}\cdot\big(d_T(t_1,r_i)+d_T(r_i,t_i)\big)=\frac{R+1}{R-1}\cdot d_T(t_1,t_i)~.
\end{align*}

Finally, we show by induction that for every pair of terminals $t_i,t_j\in K\setminus\{t_1\}$, $d_{T_k}(t_i,t_j)<\frac{\left(R+1\right)^{2}}{R-1}\cdot d_{T}(t_i,t_j)$. If $t_i,t_j$ belong to the same connected component of $T\setminus V_1$ then the argument follows by the induction hypothesis. Otherwise, $t_i\in C_i$ and $t_j\in C_j$ for $i\ne j$. 
Recall that there is a single edge $e_i=\{s_i,r_i\}$ from $C_i$ to $V_1$. Clearly, the unique path in $T$ from $t_i$ to $t_j$ goes through $V_1$ and in particular through $s_i$ and $s_j$ (note that it is possible that $s_i=s_j$). Therefore, $d_T(t_i,t_j)\ge d_T(t_i,s_i)+d_T(s_j,t_j)$. As $s_i\in V_1$, it holds that $d_{T}\left(t_{1},s_{i}\right)\le R\cdot D(s_{i})\le R\cdot d_{T}\left(t_{i},s_{i}\right)$. Therefore,
\begin{equation}
d_{T}\left(t_{1},t_{i}\right)\le d_{T}\left(t_{1},s_{i}\right)+d_{T}\left(s_{i},t_{i}\right)\le\left(R+1\right)\cdot d_{T}\left(s_{i},t_{i}\right).\label{eq:t2tibound}
\end{equation}
Similarly $d_{T}\left(t_{1},t_{j}\right)\le\left(R+1\right)\cdot d_{T}\left(s_{j},t_{j}\right)$.
Using our claim above about $t_1$, we conclude (see \Cref{fig:distortionBound} for illustration)
\begin{align*}
d_{T_{k}}\left(t_{i},t_{j}\right) & \le d_{T_{k}}\left(t_{i},t_{1}\right)+d_{T_{k}}\left(t_{1},t_{j}\right)\\
& \le\frac{R+1}{R-1}\cdot\left(d_{T}\left(t_{i},t_{1}\right)+d_{T}\left(t_{1},t_{j}\right)\right)\\
& \overset{(\ref{eq:t2tibound})}{\le}\frac{\left(R+1\right)^{2}}{R-1}\cdot\left(d_{T}\left(t_{i},s_{i}\right)+d_{T}\left(s_{j},t_{j}\right)\right)\\
& \le\frac{\left(R+1\right)^{2}}{R-1}\cdot d_{T}\left(t_{i},t_{j}\right).
\end{align*}
The expression $\frac{\left(R+1\right)^{2}}{R-1}$ is minimized by
choosing $R=3$, which proves the upper bound $8$.

\section{\MOE for Doubling Metrics}\label{sec:doubling}
In this section we analyze the \NV algorithm (Algorithm \ref{alg:metric}) for the \MOE problem, in the case where the metric spaces restricted on the terminals has doubling dimension $\ddim$.
Given a metric space $(X,d)$, Gonzalez's order \cite{G85} is defined as follows. $x_1$ is an arbitrary point, $x_2$ is the farthest point from $x_1$, and in general $x_i$ is the farthest point from $\{x_1,\dots,x_{i-1}\}$. In other words, $x_i$ is the point maximizing $d(x_i,\{x_1,\dots,x_{i-1}\})$.

\begin{theorem}\label{thm:mainDoubling}
	Let $(X,d)$ be a metric space with a set of terminals $K\subseteq X$ such that the metric space restricted to the terminals has doubling dimension $\ddim$. 
	Let $\pi$ be Gonzalez's order. Let $R_j=2\cdot e^{Z_j}$, where $Z_1,\dots,Z_k$ are \iid variables sampled according to the distribution ${\rm EXP}(c\cdot\ddim)$ for large enough constant $c$. Then the expected  distortion returned by the \NV algorithm for the \MOE problem is $O(\ddim)$.
\end{theorem}

\begin{proof}
	Consider a point $x\in X$, and let $i_{x}$ be the minimal index such that
	$d(t_{x},t_{i_{x}})\le D(x)$. Set $K_{x}=\left\{ t_{1},\dots,t_{i_{x}}\right\} $.
	As $R_{i_{x}}\ge2$, if $x$ is unassigned until the $i_x$ round, then $f(x)=t_{i_x}$. Therefore, $f(x)\in K_{x}$. For every $t,t'\in K_x\setminus \{t_{x}\}$, $d(t,t')\ge D(x)$.
	Using the packing property, for $i\ge 1$, $\left|B\left(v,2^{i}\cdot D(v)\right)\cap K_{x}\right|\le\left|B\left(t_{x},\left(2^{i}+1\right)\cdot D(v)\right)\cap K_{x}\right|=2^{O\left(i\cdot\ddim\right)}$.
	\begin{lemma}
		\label{lem:ExpectedTarget}For every $x\in X$, $\mathbb{E}\left[d(x,f(x))\right]=O\left(1\right)\cdot D(x)$.
	\end{lemma}
	
	\begin{proof}
		For $i\ge 3$, let $K_i\subseteq K_x$ be the set of terminals at distance $[2^{i-1},2^i)\cdot D(v)$ from $x$. In order for the terminal $t_j\in K_i$ to cover $x$, it must be that $R_j\ge 2^{i-1}$, where a terminal $t$ covers a point $z$ if $f(z)=t$. This happens with probability at most  
		\[
		\Pr\left[R_{j}\ge2^{i-1}\right]=\Pr\left[Z_{j}\ge(i-2)\cdot\ln2\right]=e^{-c\cdot\ddim\cdot(i-2)\cdot\ln2}\le e^{-\frac{c}{5}\cdot\ddim\cdot i}~.
		\]
		By the union bound, the probability that some terminal from $K_i$ covers $x$ is bounded by $|K_i|\cdot e^{-\frac{c}{5}\cdot\ddim\cdot i}$. We conclude that for large enough constant $c$,
		\begin{align*}
		\mathbb{E}\left[d(x,f(x))\right] & \le2^{2}\cdot D(x)+\sum_{i=3}^{\infty}\Pr\left[f(x)\in K_{i}\right]\cdot2^{i}\cdot D(x)\\
		& =4\cdot D(x)+D(x)\cdot\sum_{i=3}^{\infty}2^{O(i\cdot\ddim)}\cdot e^{-\frac{c}{5}\cdot\ddim\cdot i}\cdot2^{i}=O(D(x))~.
		\end{align*}
	\end{proof}
	
	Consider a pair of points $x,y\in X$ such that $d(x,y)=\eps\cdot\min\{D(x),D(y)\}$. If $\eps=\Omega(1)$, assume w.l.o.g that $D(x)\le D(y)$, then $D(y)\le D(x)+d(x,y)=O(1)\cdot d(x,y)$. Using \Cref{lem:ExpectedTarget} we conclude
	\begin{align}
	\mathbb{E}\left[d(f(x),f(y))\right] & \le\mathbb{E}\left[d(f(x),x)\right]+d(x,y)+\mathbb{E}\left[d(y,f(y))\right]\nonumber\\
	& =O\left(D(x)+D(y)\right)+d(x,y)=O\left(1\right)\cdot d(x,y)~.\label{eq:LargeEps}
	\end{align}
	Thus from now on we can assume that $\eps$ is upper bounded by small enough constant,
	and we also drop the assumption that $D(x)\le D(y)$. 
	We say that a terminal $t_j$ \emph{settles} the pair $\{x,y\}$ if it is the first terminal to cover at least
	one point among $\{x,y\}$, and denote this event by $\mathcal{S}_{j}$. We say that $t_{j}$ \emph{cuts} $\{x,y\}$ if $t_{j}$
	settles $\{x,y\}$ but covers only one of $x,y$, and denote this event by $\mathcal{C}_{j}$. 
	Set $R_x=\frac{d(x,t_{j})}{D(x)}$, $R_y=\frac{d(y,t_{j})}{D(y)}$. Assuming w.l.o.g that $R_x\le R_y$, we get
	\begin{equation}
	R_{y}=\frac{d(t_{j},y)}{D(y)}\le\frac{d(t_{j},x)+d(v,u)}{D(x)-d(v,u)}\le\frac{R_{x}\cdot D(x)+\epsilon\cdot D(x)}{D(x)-\epsilon\cdot D(x)}\le\frac{1+\epsilon}{1-\epsilon}\cdot R_{x}<\left(1+3\epsilon\right)\cdot R_{x}~.\label{eq:RyBound}
	\end{equation}
	Assuming that $t_j$ settles $\{x,y\}$, using the memoryless property we can bound the probability that $t_j$ cuts $\{x,y\}$.
	\begin{align}
	\Pr\left[\mathcal{C}_{j}\mid\mathcal{S}_{j}\right] & =\Pr\left[R_{j}<R_{y}\mid R_{j}\ge R_{x}\right]\overset{(\ref{eq:RyBound})}{<}\Pr\left[2\cdot e^{Z_{j}}<R_{x}\cdot(1+3\epsilon)\mid2\cdot e^{Z_{j}}<R_{x}\right]\nonumber\\
	& =\Pr\left[Z_{j}<\ln(1+3\epsilon)\right]<\Pr\left[Z_{j}<3\epsilon)\right]=1-e^{-3\epsilon\cdot c\cdot\ddim}\le6\epsilon\cdot c\cdot\ddim~.\label{eq:cutP}
	\end{align}
	Suppose that $t_j$ indeed cuts $\{x,y\}$. Following the same arguments as \Cref{lem:ExpectedTarget}, the expected distance between $y$ to $f(y)$ still will be $O(D(y))=O(\frac1\eps)\cdot  d(x,y)$. Thus,  
	\begin{equation}
	\mathbb{E}\left[d(f(x),f(y))\mid\mathcal{C}_{j}\right]\le d(t_{j},x)+d(x,y)+\mathbb{E}\left[d(y,f(y))\mid\mathcal{C}_{j}\right]=d(t_{j},\{x,y\})+O\left(\frac{1}{\epsilon}\right)\cdot d(x,y)~.\label{eq:ExpCutDist}
	\end{equation}
	For $i\ge 1$, denote by $\tilde{K}_i\subseteq K_x\cup K_y$ the set of terminals at distance $[2^{i-1},2^i)\cdot \min\{D(x),D(y)\}$ from $\{x,y\}$. By packing arguments, $|\tilde{K}_i|=2^{O(i\cdot\ddim)}$. 
	By similar arguments to \Cref{lem:ExpectedTarget}, for $i\ge 3$, the probability that $\{x,y\}$ is settled by a terminal from $\tilde{K}_i$ is bounded by
	$2^{-\Omega(i\cdot\ddim)}$.
	We conclude,
	\begin{align*}
	\mathbb{E}\left[d(f(x),f(y))\right] 
	& =\sum_{j}\Pr\left[\mathcal{S}_{j}\right]\cdot\Pr\left[\mathcal{C}_{j}\mid\mathcal{S}_{j}\right]\cdot\mathbb{E}\left[d(f(x),f(y))\mid\mathcal{C}_{j}\right]\\
	& \overset{(\ref{eq:cutP},\ref{eq:ExpCutDist})}{\le}6\epsilon\cdot c\cdot\ddim\cdot\sum_{j}\Pr\left[\mathcal{S}_{j}\right]\cdot\left(d(t_{j},\{x,y\})+O\left(\frac{1}{\epsilon}\right)\cdot d(x,y)\right)\\
	& =O\left(\ddim\right)\cdot d(x,y)+O\left(\epsilon\cdot\ddim\right)\cdot\left(4+\sum_{i\ge3}2^{-\Omega(i\cdot \ddim)}\cdot2^{i}\right)\cdot\min\{D(x),D(y)\}\\
	& =O\left(\ddim\right)\cdot d(x,y).
	\end{align*}
	
\end{proof}

\section{Connected \MOE}
The focus of this section is applying the (Graphic-)\NV algorithm (Algorithm \ref{alg:mainSPR}) for the connected-\MOE problem.

\begin{theorem}\label{thm:connected}
	Let $G=(V,E,w)$ be a weighted graph and $K\subseteq X$ a set of terminals of size $k$.
	Let $\pi$ be arbitrary, and let $R_j=e^{Z_j}$, where $Z_1,\dots,Z_k$ are \iid variables sampled according to distribution ${\rm EXP}(c\cdot\ln k)$ for large enough constant $c$. Then the expected  distortion returned by the \NV algorithm for the connected \MOE problem is $O(\log k)$.
\end{theorem}
By the triangle inequality, it is enough to prove that for every edge $\{u,v\}\in E$ (where $d_G(v,u)=w(v,u)$) it holds that $\mathbb{E}_{f\sim\mathcal{D}}\left[d(f(u),f(v)\right]\le\alpha\cdot d_G(v,u)$.
The proof itself follows almost the same lines as the proof of \Cref{thm:mainDoubling}. With high probability, $R_j\le 2$ for every terminal $t_j$. Therefore, for every vertex $v$, $d(v,f(v))\le 2\cdot D(v)$. Once a vertex $v$ joins the cluster $V_j$, the probability that its unclustered neighbor vertex $u$, at distance $\eps\cdot D(v)$, does not join $V_j$ is bounded by $O(\eps\cdot \log k)$ (similarly to \Cref{eq:cutP}). Using these two facts we can bound the expected distortion by $O(\log k)$. We skip the exact details.

\section{Linear-Time Implementation}\label{sec:time}
Our algorithm often uses $D(v)$. The next lemma state that this values can be computed efficiently.
\begin{lemma}\label{lem:SubsetTreeDijekstra}
	There is a linear-time algorithm, that given as an input a  weighted graph $G=\left(V,E,w\right)$ and $K\subseteq V$ a set of terminals, outputs for every vertex $v\in V$ its distance from $K$.
\end{lemma}

\begin{proof}
	We describe the algorithm. We root the tree in some arbitrary vertex $r\in V$. Thus each vertex (other then $r$) has a parent vertex. Our algorithm has two phases. In the first phase we sweep the tree upwards from the leafs to the root. For a vertex $v$, denote by $d(v)$ the distance from $v$ to it's closest terminal among its descendants ($\infty$ if it has no descendant terminal). The goal of the first phase is for each vertex to learn $d(v)$, and this is done in a dynamic programing fashion according to the order induced by the tree. At the beginning each leaf $v$ know $d(v)$ ($0$ if terminal and $\infty$ otherwise). Then, iteratively each internal vertex $v$ with children $\left\{ v_{1},\dots,v_{s}\right\}$ computes $d(v)=\min_{i}\left\{ d(v_{i})+d(v_{i},v)\right\}$ or  $d(v)=0$ if $v$ itself is a terminal. It is straightforward by induction that by the end of the first phase each vertex has the right value of $d(v)$. Moreover, for the root vertex $r$, $D(r)=d(r)$ (as all the terminals are the descendants of $r$). 
	
	In the second phase we sweep the tree downwards from the root to the leaves. In the first step, $r$ informs all its children the value $D(r)$. Then, iteratively, each vertex $v$ with parent $v'$ computes $D(v)=\min\left\{ d(v),D(v')+d(v',v)\right\} $. Again, by induction this is indeed the right value (as every path ending in $v$ which starts at a non-descendant of $v$ must go through $v'$).
	By the end of the second phase each vertex knows the correct value of $D(v)$. The linear time implementation follows as we traversed each edge exactly twice.
\end{proof}

The execution of the  \texttt{Relaxed-Voronoi} algorithm starts by computing the $D(v)$ values in linear time according to \Cref{lem:SubsetTreeDijekstra}. Next, in order to determine the permutation $\pi$, we choose an arbitrary vertex $r$ and run Dijkstra from it. In a tree, one can run the classic Dijkstra algorithm (as in \cite{FT87}) using a queue instead of a heap. As there is a unique path from $r$ to any other vertex, the algorithm still works properly.
Next, we cluster the vertices according to the permutation $\pi$. The set $N$ from the  \texttt{Create-Cluster} procedure can be implemented as a simple queue. As there is a unique path between every pair of vertices, once a vertex $v$ joins $N$, we can update $d(v,t_j)$ to its correct value. Moreover, there is no reason to maintain $U$.
As in all the executions of the \texttt{Create-Cluster} procedure for all terminals, each edge is traversed exactly once, the total linear time follows.

\bibliographystyle{alphaurlinit}

\bibliography{SteinerBib}

\newcommand{\etalchar}[1]{$^{#1}$}
\begin{thebibliography}{9999}

\bibitem[AFH{\etalchar{+}}04]{AFHKTT04}
A.~Archer, J.~Fakcharoenphol, C.~Harrelson, R.~Krauthgamer, K.~Talwar, and
  {\'{E}}.~Tardos.
\newblock Approximate classification via earthmover metrics.
\newblock In {\em Proceedings of the Fifteenth Annual {ACM-SIAM} Symposium on
  Discrete Algorithms, {SODA} 2004}, pages 1079--1087, 2004.
\newblock \href
  {http://dx.doi.org/http://dl.acm.org/citation.cfm?id=982792.982952}
  {\path{doi:http://dl.acm.org/citation.cfm?id=982792.982952}}.

\bibitem[BG08]{BG08}
A.~Basu and A.~Gupta.
\newblock Steiner point removal in graph metrics.
\newblock Unpublished Manuscript, available from
  \url{http://www.math.ucdavis.edu/~abasu/papers/SPR.pdf}, 2008.

\bibitem[Che18]{Che18}
Y.~K. Cheung.
\newblock Steiner point removal - distant terminals don't (really) bother.
\newblock In {\em Proceedings of the Twenty-Ninth Annual {ACM-SIAM} Symposium
  on Discrete Algorithms, {SODA} 2018}, pages 1353--1360, 2018.
\newblock \href {http://dx.doi.org/10.1137/1.9781611975031.89}
  {\path{doi:10.1137/1.9781611975031.89}}.

\bibitem[CKR04]{CKR04}
G.~C{\u{a}}linescu, H.~J. Karloff, and Y.~Rabani.
\newblock Approximation algorithms for the 0-extension problem.
\newblock {\em {SIAM} J. Comput.}, 34(2):358--372, 2004.
\newblock \href {http://dx.doi.org/10.1137/S0097539701395978}
  {\path{doi:10.1137/S0097539701395978}}.

\bibitem[CXKR06]{CXKR06}
T.-H. Chan, D.~Xia, G.~Konjevod, and A.~Richa.
\newblock A tight lower bound for the steiner point removal problem on trees.
\newblock In {\em Proceedings of the 9th International Conference on
  Approximation Algorithms for Combinatorial Optimization Problems, and 10th
  International Conference on Randomization and Computation},
  APPROX'06/RANDOM'06, pages 70--81, 2006.
\newblock \href {http://dx.doi.org/10.1007/11830924_9}
  {\path{doi:10.1007/11830924_9}}.

\bibitem[DJP{\etalchar{+}}92]{DJPSY92}
E.~Dahlhaus, D.~S. Johnson, C.~H. Papadimitriou, P.~D. Seymour, and
  M.~Yannakakis.
\newblock The complexity of multiway cuts (extended abstract).
\newblock In {\em Proceedings of the 24th Annual {ACM} Symposium on Theory of
  Computing, STOC 1992}, pages 241--251, 1992.
\newblock \href {http://dx.doi.org/10.1145/129712.129736}
  {\path{doi:10.1145/129712.129736}}.

\bibitem[EGK{\etalchar{+}}14]{EGKRTT14}
M.~Englert, A.~Gupta, R.~Krauthgamer, H.~R{\"{a}}cke, I.~Talgam{-}Cohen, and
  K.~Talwar.
\newblock Vertex sparsifiers: New results from old techniques.
\newblock {\em {SIAM} J. Comput.}, 43(4):1239--1262, 2014.
\newblock \href {http://dx.doi.org/10.1137/130908440}
  {\path{doi:10.1137/130908440}}.

\bibitem[FHRT03]{FHRT03}
J.~Fakcharoenphol, C.~Harrelson, S.~Rao, and K.~Talwar.
\newblock An improved approximation algorithm for the 0-extension problem.
\newblock In {\em Proceedings of the Fourteenth Annual {ACM-SIAM} Symposium on
  Discrete Algorithms ({SODA} 2003)}, pages 257--265, 2003.
\newblock Available from:
  \url{http://dl.acm.org/citation.cfm?id=644108.644153}.

\bibitem[Fil18a]{Fil18}
A.~Filtser.
\newblock Steiner point removal with distortion \emph{O}(log \emph{k}).
\newblock In {\em Proceedings of the Twenty-Ninth Annual {ACM-SIAM} Symposium
  on Discrete Algorithms, {SODA} 2018}, pages 1361--1373, 2018.
\newblock \href {http://dx.doi.org/10.1137/1.9781611975031.90}
  {\path{doi:10.1137/1.9781611975031.90}}.

\bibitem[Fil18b]{Fil18a}
A.~Filtser.
\newblock Steiner point removal with distortion \emph{O}(log \emph{k}), using
  the noisy-voronoi algorithm.
\newblock {\em CoRR}, abs/1808.02800, 2018.
\newblock \href {http://arxiv.org/abs/1808.02800} {\path{arXiv:1808.02800}}.

\bibitem[FT87]{FT87}
M.~L. Fredman and R.~E. Tarjan.
\newblock Fibonacci heaps and their uses in improved network optimization
  algorithms.
\newblock {\em J. {ACM}}, 34(3):596--615, 1987.
\newblock \href {http://dx.doi.org/10.1145/28869.28874}
  {\path{doi:10.1145/28869.28874}}.

\bibitem[GKL03]{GKL03}
A.~Gupta, R.~Krauthgamer, and J.~R. Lee.
\newblock Bounded geometries, fractals, and low-distortion embeddings.
\newblock In {\em 44th Symposium on Foundations of Computer Science {(FOCS}
  2003)}, pages 534--543, 2003.
\newblock \href {http://dx.doi.org/10.1109/SFCS.2003.1238226}
  {\path{doi:10.1109/SFCS.2003.1238226}}.

\bibitem[Gon85]{G85}
T.~F. Gonzalez.
\newblock Clustering to minimize the maximum intercluster distance.
\newblock {\em Theor. Comput. Sci.}, 38:293--306, 1985.
\newblock \href {http://dx.doi.org/10.1016/0304-3975(85)90224-5}
  {\path{doi:10.1016/0304-3975(85)90224-5}}.

\bibitem[GT13]{GT13}
A.~Gupta and K.~Talwar.
\newblock Random rates for 0-extension and low-diameter decompositions.
\newblock {\em CoRR}, abs/1307.5582, 2013.
\newblock \href {http://arxiv.org/abs/1307.5582} {\path{arXiv:1307.5582}}.

\bibitem[Gup01]{G01}
A.~Gupta.
\newblock Steiner points in tree metrics don't (really) help.
\newblock In {\em Proceedings of the Twelfth Annual ACM-SIAM Symposium on
  Discrete Algorithms}, SODA '01, pages 220--227, 2001.
\newblock Available from:
  \url{http://dl.acm.org/citation.cfm?id=365411.365448}.

\bibitem[Kar98]{Kar98}
A.~V. Karzanov.
\newblock Minimum 0-extensions of graph metrics.
\newblock {\em Eur. J. Comb.}, 19(1):71--101, 1998.
\newblock \href {http://dx.doi.org/10.1006/eujc.1997.0154}
  {\path{doi:10.1006/eujc.1997.0154}}.

\bibitem[KKN15]{KKN15}
L.~Kamma, R.~Krauthgamer, and H.~L. Nguyen.
\newblock Cutting corners cheaply, or how to remove steiner points.
\newblock {\em {SIAM} J. Comput.}, 44(4):975--995, 2015.
\newblock \href {http://dx.doi.org/10.1137/140951382}
  {\path{doi:10.1137/140951382}}.

\bibitem[KT02]{KT02}
J.~M. Kleinberg and {\'{E}}.~Tardos.
\newblock Approximation algorithms for classification problems with pairwise
  relationships: metric labeling and markov random fields.
\newblock {\em J. {ACM}}, 49(5):616--639, 2002.
\newblock \href {http://dx.doi.org/10.1145/585265.585268}
  {\path{doi:10.1145/585265.585268}}.

\bibitem[LN03]{LN03}
J.~R. Lee and A.~Naor.
\newblock Metric decomposition, smooth measures, and clustering.
\newblock Unpublished Manuscript, available from
  \url{https://www.math.nyu.edu/~naor/homepage\%20files/cluster.pdf}, 2003.

\bibitem[MPX13]{MPX13}
G.~L. Miller, R.~Peng, and S.~C. Xu.
\newblock Parallel graph decompositions using random shifts.
\newblock In {\em 25th {ACM} Symposium on Parallelism in Algorithms and
  Architectures, {SPAA} '13}, pages 196--203, 2013.
\newblock \href {http://dx.doi.org/10.1145/2486159.2486180}
  {\path{doi:10.1145/2486159.2486180}}.

\end{thebibliography}

\end{document}